\documentclass[11pt]{article}
\usepackage{amssymb,amsmath,amsthm,amscd,latexsym}
\usepackage{mathrsfs}
\usepackage{mathrsfs}
\usepackage{amsfonts}
\usepackage{amsmath}
\usepackage{amssymb}
\usepackage{amscd}

\renewcommand{\paragraph}{\roman{paragraph}}
\input cyracc.def

\newcommand{\F}{\mathbb{F}}

\newtheorem{thm}{\scshape \mdseries  Theorem}[section]
\newtheorem{lem}[thm]{\scshape \mdseries  Lemma}

\newtheorem{ex}[thm]{\scshape \mdseries  Example}

\begin{document}
\title{\bf On self-dual and LCD \\ double circulant and double negacirculant codes\\ over $\F_q+u\F_q$
\thanks{This research is supported by National Natural Science Foundation of China (61672036), Technology Foundation for Selected Overseas Chinese Scholar, Ministry of Personnel of China (05015133) and the Open Research Fund of National Mobile Communications Research Laboratory, Southeast University (2015D11) and Key projects of support program for outstanding young talents in Colleges and Universities (gxyqZD2016008).}
}
\author{
Minjia Shi,\thanks{ Minjia Shi, Key Laboratory of Intelligent Computing Signal Processing, Ministry of Education, Anhui University, No.3 Feixi Road, Hefei, Anhui, 230039, China, School of Mathematical Sciences, Anhui University, Hefei, Anhui, 230601,
China and National Mobile Communications Research Laboratory, Southeast University, China, {\tt smjwcl.good@163.com}}
Hongwei Zhu,\thanks{Hongwei Zhu, School of Mathematical Sciences, Anhui University, Hefei, Anhui, 230601, {\tt zhwgood66@163.com}}
Liqin Qian,\thanks{Liqin Qian, School of Mathematical Sciences, Anhui University, Hefei, Anhui, 230601, {\tt qianliqin\_1108@163.com}}
Lin Sok,\thanks{Lin Sok, School of Mathematical Sciences, Anhui University, Hefei, Anhui, 230601 and Department of Mathematics, Royal University of Phnom Penh, 12156 Phnom Penh, Cambodia, { \tt sok.lin@rupp.edu.kh}}
and
Patrick Sol\'e\thanks{Patrick Sol\'e, CNRS/LAGA, University of Paris 8, 93 526 Saint-Denis, France, {\tt sole@enst.fr}}
}

\date{}
\maketitle
\begin{abstract}
Double circulant codes of length $2n$ over the semilocal ring $R=\F_q+u\F_q,\, u^2=u,$ are studied when $q$ is an odd prime power, and $-1$ is a square in $\F_q.$
Double negacirculant codes of length $2n$ are studied over $R$
  when $n$ is even and $q$ is an odd prime power. Exact enumeration of self-dual and LCD such codes for given length $2n$ is given.
  Employing a duality-preserving Gray map, self-dual and LCD codes of length $4n$ over $\F_q$ are constructed.
  Using random coding and the Artin conjecture, the relative distance of these codes is bounded below. The parameters of examples of the modest length are computed.
  Several such codes are optimal.
\end{abstract}

\section{Introduction}
Double circulant and double negacirculant self-dual codes over finite fields have been studied recently in \cite{A} and \cite{AOS}, respectively, from the viewpoint of enumeration and asymptotic performance. The main tool there is the CRT approach to quasi-cyclic codes as introduced in \cite{LS1}, and generalized to quasi-twisted codes in \cite{Y}. It is natural to extend these results to other classes of ring alphabets, beyond finite fields, building on the theory developed in \cite{LS2}. Double circulant self-dual codes over a commutative ring can only exist if there is a square root of $-1$ over that ring \cite{LS2}. For the ring $R=\F_q+u\F_q,\, u^2=u,$ a ring of current interest \cite{SGS}, this requirement leads to the condition that $-1$ is a square in $\F_q.$ This ring is studied here because of a duality preserving Gray map that turns self-dual codes into self-dual codes, and similarly, LCD codes into LCD codes. While this Gray map was defined already in \cite{ZW}, its duality properties were not considered there. LCD codes are defined as codes intersecting trivially with their dual. They enjoy a wealth of combinatorial and algebraic properties \cite{LCD,QCCD}. In particular quasi-cyclic codes over fields have been proved to contain infinite families of good long LCD codes \cite{QCCD}. In this paper we study both
self-dual and LCD double circulant and double negacirculant codes over $R$ and show that all four families contain arbitrarily long codes with fixed rate and with relative distance of the Gray image bounded below by a nonzero constant. These results are based on the complete enumeration of the codes in these four families for given length $2n.$ While, for asymptotic purposes, only the case of $n$ a prime with $q$ primitive modulo $n$ is needed, the counting results are developed for all $n$'s. The said arithmetic condition holds for infinitely many $n$'s by Artin conjecture \cite{M}, which is proved conditionally under GRH \cite{H}. In the case of negacirculant codes, $n$ is supposed to be a power of $2$ for asymptotics, and the truth of Artin conjecture is not needed. The factorization of $x^n+1$ in that case depends on properties of Dickson polynomials.

The material is organized as follows. Section $2$ contains the preliminary requisite necessary to the further sections. Section $3$ develops the machinery of the CRT approach to double circulant and double negacirculant codes, and derives the main enumeration results. Section $4$ is dedicated to asymptotic bounds on the relative Hamming distance of the Gray image of these codes when their lengths tends to infinity. Section $5$ computes some numerical examples.
Section $6$ concludes the article and discusses open problems.
\section{\textbf{Preliminaries}}
\subsection{\textbf{The Ring $\F_q+u\F_q$}}
Let $q$ be a prime power, and let $\mathbb{F}_q$ be the finite field of order $q.$ Consider the ring $R=\F_q+u\F_q$ where $u^2=u$.
It is semi-local with maximal ideals $(u)$ and $(u-1)$. The group of units $R^*$ is isomorphic, as a multiplicative group, to the product of two
cyclic groups of order $(q-1)$. To construct double circulant codes, we need to know when the ring $R=\F_q+u\F_q$ contains a square root of $-1$.

{\thm The ring $R$ contains a square root of $-1$ if one and only one of the following three conditions is true:
 \begin{itemize}
\item $q$ is a power of $2$;
\item $q=p^b$, where $p$ is a prime congruent to $1\mod 4$;
\item $q=p^{2b}$, where $p$ is a prime congruent to $3\mod 4,$ and $b\ge 1$ is an integer.
\end{itemize}
}

\begin{proof}
Note first that $\F_q$ contains a square root of $-1$ iff one of the above three conditions is true \cite{LS2}.
The condition is sufficient.
If one of the three cases is true, then $\F_q,$ a subring of $R,$ already contains such a root. The condition is necessary.
Write such a root as $x+uy$ with $x,y \in \F_q$ . Then $(x+uy)^2=x^2+u(y^2+2xy),$ showing that $x^2=-1,$
and that $x$ is a square root of $-1$ in $\F_q.$
\end{proof}

{\bf Remark:} Note that for each $x$ in the above proof there are two possible value of $y$ that is $y=0,$ or $y=-2x.$ There are thus zero or four distinct square roots of $-1$ in $R.$
\subsection{\textbf{ Norm Function over Finite Field}}
For all $x\in\F_{q^n}$, the norm of $x$ over $\F_q$ is a map {\bf Norm}: $\F_{q^n}\longrightarrow\F_q$ defined by
$$Norm(x)=x^{(q^n-1)/(q-1)}.$$
Moreover, Norm is a multiplicative homomorphism which is surjective (\cite[Theorem 2.28]{FF}). $Norm(0)=0$, so it maps $\F^*_{q^n}$ onto $\F^*_{q}$, where each nonzero element in $\F^*_q$ has a preimage of size $(q^n-1)/(q-1)$ in $\F^*_{q^n}.$
\subsection{\textbf{Codes}}
A {\bf linear code} $C$ over $R$ of length $n$ is an $R$-submodule of $R^n$. If $x=(x_1,x_2,\ldots,x_n)$ and $y=(y_1,y_2,\ldots,y_n)$ are two elements of $R^n$, their standard (Euclidean) inner product is defined by $\langle x,y\rangle=\sum\limits_{i=1}^{n}x_iy_i$, where the operation is performed in $R$. The Euclidean dual code of $C$ is denoted by $C^{\perp}$ and defined as $C^{\perp}=\{y\in R^n\mid\langle x,y\rangle=0,\forall x\in C\}$.
A linear code $C$ of length $n$ over $R$ is called {\bf self-dual} if $C=C^{\perp}$.
A linear code $C$ of length $n$ over $R$ is called {\bf linear complementary dual (LCD)} if $C \cap C^{\perp}=\{\mathbf{0}\}.$

 A matrix $A$ over $R$ is said to be {\bf circulant} ($resp.$ {\bf negacirculant}) if its rows are obtained by successive shifts ($resp.$ {\bf negashifts}) from the first row. A code $C$ is called {\bf double circulant} ($resp.$ {\bf double negacirculant}) over $R$ if its generator matrix $G$ will be of the form $G=(I,A)$, where $I$ is the identity matrix of order $n$ and $A$ is a circulant ($resp.$ negacirculant) matrix of the same order.

If $C(n)$ is a family of codes with parameters $[n,k_n,d_n]$ over $\F_q$, the rate $\rho$ and relative distance $\delta$ are defined as
$\rho=\limsup\limits_{n\rightarrow \infty}\frac{k_n}{n}$
and
$\delta=\liminf\limits_{n\rightarrow \infty}\frac{d_n}{n}$, respectively.
A family of code is {\bf good} if $\rho \delta >0.$
Both limits are finite as limits of bounded quantities.
\subsection{\textbf{Gray Map}}
\subsubsection{$q$ Odd}
The Gray map $\phi$ over the ring $R$ is defined by $\phi(a+ub)=(-b,2a+b)$ for $a,b\in\F_q$ from $R$ to $\F^2_q$. Clearly, $\phi$ is a bijection from $R$ to $\F^2_q$, which extends naturally to a map from $R^n$ to $\F_q^{2n}$. Adopting the idea of \cite[Theorem 2.10.3]{LCD}, we can prove the following result.

\begin{thm}\label{a12}
If $C$ is a self-dual ($resp.$ LCD) code of length $n$ over $R$, then $\phi(C)$ is a self-dual ($resp.$ LCD) code of length $2n$ over $\F_{q}.$
\end{thm}
\begin{proof}
Consider a pair of codewords $w=(w_1,w_2,\ldots,w_n)$, $v=(v_1,v_2,\ldots,v_n)\in C$, where $w_i=a_i+b_iu,$ $v_i=c_i+d_iu$ and
$a_i,b_i,c_i,d_i\in \F_q$ for $1\leq i\leq n.$ If $C$ is a self-dual code over $R$, then
$\langle w,v\rangle=\sum\limits_{i=1}^n(a_i+b_iu)(c_i+d_iu)=\sum\limits_{i=1}^{n}(a_ic_i+(a_id_i+b_ic_i+b_id_i)u)=0$, which implies that
$$\sum\limits_{i=1}^na_ic_i=0 \ {\rm{and}} \ \ \sum\limits_{i=1}^n(a_id_i+b_ic_i+b_id_i)=0.$$

According to the definition of the Gray map, we have
\begin{eqnarray*}
  \langle\phi(w),\phi(v)\rangle &=&\langle\phi(a_1+b_1u,\ldots,a_n+b_nu),\phi(c_1+d_1u,\ldots,c_n+d_nu)\rangle\\
  &=&\sum\limits_{i=1}^n(b_id_i+4a_ic_i+2a_id_i+2b_ic_i+b_id_i)\\
  &=&\sum\limits_{i=1}^n(4a_ic_i+2(a_id_i+b_ic_i+b_id_i))=0.\\
\end{eqnarray*} It implies that $\phi(C^\perp)\subseteq\phi(C)^{\perp}$. Since the Gray map $\phi$ is a bijection from $R$ to $\F_q^2$, then $\phi(C^{\perp})=\phi(C)^{\perp}$.

If $C$ is a LCD code over $R$, then $C\cap C^{\perp}=\{\mathbf{0}\}$.
We can easily obtain that $\phi(C\cap C^{\perp})\subseteq \phi(C)\cap \phi(C^{\perp})$. Because $\phi$ is a bijection from $R$ to $\F_q^2$, equality holds in the preceding inclusion and
$$\phi(C)\cap \phi(C)^{\perp}= \phi(C)\cap \phi(C^{\perp})=\phi(C\cap C^{\perp})=\{\mathbf{0}\}.$$
Thus $\phi(C)$ is a LCD code of length $2n$ over $\F_q.$
\end{proof}

\subsubsection{$q$ Arbitrary}

We can generalize the map $\beta$ of \cite{D++} as follows.
Define a map $\beta: R \rightarrow \F_q^2$ by the formula $\beta(a+ub)=(a,a+b),$ for all $a,b\in \F_q.$
\begin{thm}\label{a13}
If $C$ is a self-dual ($resp.$ LCD) code of length $n$ over $R$, then $\beta(C)$ is a self-dual ($resp.$ LCD) code of length $2n$ over $\F_{q}.$
\end{thm}

\begin{proof}
We claim that $\beta(C^\bot)=\beta(C)^\bot.$ The rest of the proof follows in the same way as the proof of Theorem \ref{a12}. To prove the claim note that
if $\langle w,v\rangle=0$ like in the said proof, then
$$\langle\beta(w),\beta(v)\rangle=\sum_{i=1}^n a_ic_i+(a_i+b_i)(c_i+d_i)=\sum_{i=1}^n a_id_i+b_ic_i+b_id_i=0.$$
This shows that $\beta(C^\bot)\subseteq \beta(C)^\bot,$ and $\beta$ being bijective, that $\beta(C^\bot)=\beta(C)^\bot.$
\end{proof}

\section{\textbf{ Algebraic Structure of Double Circulant and Double Negacirculant Codes}}
\subsection{\textbf{Double Circulant Codes of Odd Length}}
In this subsection, we assume that $n$ is an odd integer, $q$ is a prime power and $\gcd(n,q)=1$. We can cast the factorization of $x^n-1$ into distinct irreducible polynomials over $R$ in the form
$$x^n-1=\alpha(x-1)\prod\limits_{i=2}^sg_i(x)\prod\limits_{j=1}^th_j(x)h_j^*(x),$$
where $\alpha\in R^*,$ $g_i(x)$ is a self-reciprocal polynomial with degree $2e_i$ for $2\leq i\leq s$, and $h_j^*(x)$ is the reciprocal polynomial of $h_j(x)$ with degree $d_j$ for $1\leq j\leq t$.
By the Chinese Remainder Theorem (CRT), we have
\begin{eqnarray*}
  \frac{R[x]}{(x^n-1)} &\simeq&\frac{R[x]}{(x-1)}\oplus \left(\bigoplus_{i=2}^s R[x]/( g_i(x))\right)\oplus\left(\bigoplus_{j=1}^t(R[x]/( h_j(x))\oplus(R[x]/( h_j^*(x))))\right)\\
  &\simeq&R\oplus\left(\bigoplus_{i=2}^s \frac{\F_q[u,x]}{(u^2-u,g_i(x))}\right)\oplus\left(\bigoplus_{j=1}^t\frac{\F_q[u,x]}{(u^2-u,h_j(x))}
  \oplus\frac{\F_q[u,x]}{(u^2-u,h_j^*(x))}\right)\\
  &\simeq&R\oplus\left(\bigoplus_{i=2}^s \F_{q^{2e_i}}+u\F_{q^{2e_i}}\right)\oplus\left(\bigoplus_{j=1}^t(\F_{q^{d_j}}+u\F_{q^{d_j}})
  \oplus(\F_{q^{d_j}}+u\F_{q^{d_j}})\right)\\
  &:=&R\oplus\left(\bigoplus_{i=2}^s R_{2e_i}\right)\oplus\left(\bigoplus_{j=1}^t R_{d_j}
  \oplus R_{d_j}\right),\\
\end{eqnarray*}
where we let $R_\ell =\F_{q^{\ell}}+u\F_{q^{\ell}},$ for any nonegative integer $\ell.$
Note that all of these rings are extensions of $R$. This decomposition naturally extends to $\left({\frac{R[x]}{(x^n-1)}}\right)^2$ as
\begin{eqnarray*}
  \left(\frac{R[x]}{(x^n-1)}\right)^2 &\simeq&R^2\oplus\left(\bigoplus_{i=2}^s ({R_{2e_i}})^2\right)\oplus\left(\bigoplus_{j=1}^t ({R_{d_j}})^2
  \oplus ({R_{d_j}})^2\right).\\
\end{eqnarray*}
In particular, each linear code $C$ of length $2$ over $\frac{R[x]}{(x^n-1)}$ can be decomposed as the ``CRT sum"
$$C\simeq C_1\oplus\left(\bigoplus_{i=2}^s{C_i}\right)\oplus\left(\bigoplus_{j=1}^t({C_j^\prime}\oplus{C_j^{\prime\prime}})\right),$$
where $C_1$ is a linear code over $R$, for each $2\leq i\leq s$, $C_i$ is a linear code over $R_{2e_i}$, and for each $1\leq j\leq t$, $C_j^\prime$ and $C_j^{\prime\prime}$ are linear codes over $R_{d_j}$, which are called the constituents of $C$.

In self-reciprocal case, we give the definition of the Hermitian inner product over $R_{2e_i}$. For all $z=z_1+uz_2\in R_{2e_i}$, where $z_1,z_2\in\F_{q^{2e_i}}$, the conjugate of $z$ over $R_{2e_i}$
is $\overline{z}=z_1^{q^{e_i}}+uz_2^{q^{e_i}}$, and the Hermitian inner product is defined as: $(z,z^{\prime})\cdot(w,w^{\prime})=z\overline{w}+z^{\prime}\overline{w^{\prime}}$, where $z, z^{\prime}, w, w^{\prime}\in R_{2e_i}.$
 \begin{thm}\label{a1}
Let $n$ denote a positive odd integer, and let $q$ be a prime power coprime with $n$. Assume that the factorization of $x^n-1$ into irreducible polynomials over $R$ is of the form
$$x^n-1=\alpha(x-1)\prod\limits_{i=2}^sg_i(x)\prod\limits_{j=1}^th_j(x)h_j^*(x),$$
with $\alpha\in R^*$, $n=1+\sum\limits_{i=2}^s2e_i+2\sum\limits_{j=1}^td_j$. Then the total number of self-dual double circulant codes over $R$ is
$$4\prod\limits_{i=2}^{s}\left(1+q^{e_i}\right)^2\prod\limits_{j=1}^t\left(q^{d_j}-1\right)^2.$$
\end{thm}
\begin{proof}
We can count the number of self-dual double circulant codes by counting their constituent codes. There are four self-dual codes $C_1$ of length $2$ over $R$,
whose generators are $(1, \omega)$, $(1, -\omega)$, $(1, \omega(1-2u))$ and $(1,\omega(2u-1))$, respectively, where $\omega^2=-1$ and $\omega\in \F_q$.
More generally, a factor $g_i(x)$ of degree $2e_i$ leads to counting self-dual Hermitian codes $C_i$ of length $2$ over $R_{2e_i}.$ If $(1,c_{e_i})$ is the generator of $C_i$
then $(1,c_{e_i})\cdot(1,c_{e_i})=1+c_{e_i}\overline{c_{e_i}}=1+c_{e_i}{c_{e_i}}^{q^{e_i}}=0$. Let $c_{e_i}=x+uy$, where $x,y\in\F_{q^{2e_i}}$, we then have
 \begin{equation*}\label{den1}
 \small
 1+(x+uy)(x^{q^{e_i}}+uy^{q^{e_i}})=0
\Longleftrightarrow
\begin{cases}
 \emph{ }x\cdot x^{q^{e_i}}=-1, \\
   \emph{ }(x+y)(x+y)^{q^{e_i}}=-1,\\
\end{cases}\Longleftrightarrow
\begin{cases}
 \emph{ }Norm(x)=-1, \\
   \emph{ }Norm(x+y)=-1,\\
\end{cases}
\end{equation*}
where  $Norm$ is a map from $\F_{q^{2e_i}}$ to $\F_{q^{e_i}}$. So there are $q^{e_i}+1$ roots for $Norm(x)=-1$ and $q^{e_i}+1$
roots for $Norm(x+y)=-1$. Therefore, the number of ordered pairs $(x,y)$ is equal to $(q^{e_i}+1)^2$.

In the last case about reciprocal pairs, note that $h_j(x)$ and $h^*_j(x)$ of both degree $d_j$ leads to counting dual pairs of codes (for the Euclidean inner product) of length $2$ over $R_{d_j}$, that is to count the number of solutions of the equation $1+c_{d_j}^{\prime}c_{d_j}^{\prime\prime}=0$, where $(1,c_{d_j}^{\prime})$ and $(1,c^{\prime\prime}_{d_j})$ are the generators of $C_{j}^{\prime}$ and $C_{j}^{\prime\prime}$, respectively.
If $c_{d_j}^{\prime}\in R_{d_j}^*$, then $c_{d_j}^{\prime\prime}=-\frac{1}{c_{d_j}^{\prime}}$. There are $|R_{d_j}^*|=(q^{d_j}-1)^2$ choices for
$(c_{d_j}^{\prime}, c_{d_j}^{\prime\prime}).$ If $c_{d_j}^{\prime}\in R_{d_j}\backslash R_{d_j}^*$, then $c_{d_j}^{\prime}=ux\in (u)$ or $c_{d_j}^{\prime}=(u-1)x\in (u-1)$.
In this case, we claim that $1+c_{d_j}^{\prime}c_{d_j}^{\prime\prime}=0$ cannot occur. Otherwise,
by  reducing the equation modulo $u$ or $u-1$, we would get $1=0$ in $\F_{q^{d_j}}$, contradiction. The proof of the theorem is now completed.
\end{proof}
 \begin{thm}\label{a1}
Keep the same notations as above, then the total number of LCD double circulant codes over $R$ is
$$(q^2-4)\prod\limits_{i=2}^{s}(q^{4e_i}-(q^{e_i}+1)^2)\prod\limits_{j=1}^t(q^{4d_j}-2q^{3d_j}+3q^{2d_j}-2q^{d_j}+1).$$
\end{thm}
\begin{proof}
Like in the self-dual case, we can count the number of LCD double circulant codes $C$ by considering each constituent of $C$. Note that if a double circulant code of length 2
over some extension of $R$ is not self-dual, then it is LCD
 because of having single-row generator matrix. Combined with the results of Theorem 3.1, the total number of LCD double circulant
constituent codes $C_1$ over $R$ is $q^2-4$ and the total number of LCD double circulant constituent codes $C_i$ over $R_{2e_i}$ is $q^{4e_i}-(q^{e_i}+1)^2$.

In the last case about reciprocal pairs, note that $h_j(x)$ and $h^*_j(x)$ of degree $d_j$ leads to counting pairs of codes of length $2$ over $R_{d_j}$ with
certain intersection properties \cite{QCCD}. Let $(1,a)$ and $(1,b)$ be the generators of $C_j^{\prime}$ and $C_j^{\prime\prime}$, respectively.
\begin{equation*}\label{den1}
 \small
\begin{cases}
 \emph{ }C_j^{\prime}\cap{C_j^{\prime\prime}}^{\perp}=\{0\}, \\
   \emph{ }C_j^{\prime\prime}\cap {C_j^{\prime}}^{\perp}=\{0\}.\\
\end{cases}\Longleftrightarrow
1+ab\in R_{d_j}^*.
\end{equation*}

Without loss of generality, we discuss on the unit character of $a$ as follows.
\begin{itemize}
\item If $a\in R^*_{d_j}$, then $b\in \frac{-1}{a}+R^*_{d_j}$ and $|\frac{-1}{a}+R^*_{d_j}|=|R^*_{d_j}|$. In this case, we have $|R^*_{d_j}|^2=(q^{d_j}-1)^4$ choices for $(a,b)$.

\item If $a\in R_{d_j}\backslash\{ R_{d_j}^*\cup\{0\}\}$, then $a=u\alpha$ or $a=(u-1)\alpha$. Let $b=\beta^{\prime}+u\beta^{\prime\prime}$, where $\alpha\in\F^*_{q^{d_j}}$ and
$ \beta^{\prime}, \beta^{\prime\prime}\in\F_{q^{d_j}}.$ When $a=u\alpha$, we then have
$$1+ab=1+u\alpha(\beta^{\prime}+u\beta^{\prime\prime})=1+u\alpha(\beta^{\prime}+\beta^{\prime\prime})\in R^*_{d_j}.$$
This is equivalent to $\alpha(\beta^{\prime}+\beta^{\prime\prime})\neq -1$. There are $q^{d_j}(q^{d_j}-1)^2$ choices for $(\alpha,\beta^{\prime}, \beta^{\prime\prime}).$
When $a=(u-1)\alpha$, then we have $$1+ab=1+(u-1)\alpha(\beta^{\prime}+u\beta^{\prime\prime})=1+(u-1)\alpha\beta^{\prime}\in R^*_{d_j}.$$
This is equivalent to $\alpha\beta^{\prime}\neq 1$ and $\beta^{\prime\prime}$ is arbitrary in $\F_{q^{d_j}}$.
There are also $q^{d_j}(q^{d_j}-1)^2$ choices for $(\alpha,\beta^{\prime}, \beta^{\prime\prime}).$

\item
If $a$ is zero, then $b$ is arbitrary in $R_{d_j}$. There are $q^{2d_j}$ choices for $b$.
\end{itemize}
 Hence, the number of the last case about reciprocal pairs is
 $(q^{d_j}-1)^4+q^{d_j}(q^{d_j}-1)^2+q^{d_j}(q^{d_j}-1)^2+q^{2d_j}
=q^{4d_j}-2q^{3d_j}+3q^{2d_j}-2q^{d_j}+1.$ The proof of the theorem is now completed.
\end{proof}
\subsection{\textbf{Double Negacirculant Codes of Even Length}}
For our purpose, we assume in this subsection that $n$ is an even integer, $q$ is an odd prime power, and that $\gcd(n,q)=1$  . We can cast the factorization of $x^n+1$ into distinct irreducible polynomials over $R$ in the form
$$x^n+1=\alpha\prod\limits_{i=1}^sg_i(x)\prod\limits_{j=1}^th_j(x)h_j^*(x),$$
where $\alpha\in R^*,$ $g_i(x)$ is a self-reciprocal polynomial with degree $2e_i$ for $1\leq i\leq s$, and $h_j^*(x)$ is the reciprocal polynomial of $h_j(x)$ with degree $d_j$ for $1\leq j\leq t$.
Using the same notations and arguments as above, we can easily carry out the results as follows:
 $$\frac{R[x]}{(x^n+1)}\simeq \left(\bigoplus_{i=1}^s R_{2e_i}\right)\oplus\left(\bigoplus_{j=1}^t R_{d_j}
  \oplus R_{d_j}\right),$$ and
  $$C\simeq \left(\bigoplus_{i=1}^s{C_i}\right)\oplus\left(\bigoplus_{j=1}^t({C_j^\prime}\oplus{C_j^{\prime\prime}})\right).$$

Similar to the proof of Theorem 3.1 and Theorem 3.2, we have the following two enumeration results. Their proofs are omitted.

 \begin{thm}\label{a1}
Let $n$ denote a positive even integer, and $q$ a prime power coprime with $n$. Assume that the factorization of $x^n+1$ into irreducible polynomials over $R=\F_q+u\F_q$ is of the form
$$x^n+1=\alpha\prod\limits_{i=1}^sg_i(x)\prod\limits_{j=1}^th_j(x)h_j^*(x),$$
with $\alpha\in R^*$, $n=\sum\limits_{i=1}^s2e_i+2\sum\limits_{j=1}^td_j$. Then the total number of self-dual double negacirculant codes over $R$ is
$$\prod\limits_{i=1}^{s}\left(1+q^{e_i}\right)^2\prod\limits_{j=1}^t(q^{d_j}-1)^2.$$
\end{thm}

 \begin{thm}\label{a1}
Keep the same notations as above, then the total number of LCD double negacirculant codes over $R$ is
$$\prod\limits_{i=1}^{s}(q^{4e_i}-(q^{e_i}+1)^2)\prod\limits_{j=1}^t(q^{4d_j}-2q^{3d_j}+3q^{2d_j}-2q^{d_j}+1).$$
\end{thm}
\section{\textbf{ Distance Bound}}
\subsection{\textbf{ Distance Bound for Double Circulant Codes}}
Let $q$ be a primitive root mod $n$ and $n$ be an odd prime. Recall that $x^n-1=(x-1)(x^{n-1}+\ldots+x+1)=(x-1)h(x).$ Note that $h(x)$ is irreducible over $R,$ because
$\F_q$ is a subring of $R,$ and $h(x)$ is irreducible over $\F_q$.

 By the Chinese Remainder Theorem (CRT), we have
 \begin{eqnarray*}
  \frac{R[x]}{(x^n-1)}\simeq\frac{R[x]}{(x-1)}\oplus \frac{R[x]}{(h(x))}\simeq R\oplus\frac{\F_q[u,x]}{(u^2-u, h(x))}\simeq R\oplus\F_{q^{n-1}}+u\F_{q^{n-1}}.
\end{eqnarray*}
Let $\mathcal{R}$ denote the ring $\frac{R[x]}{(h(x))}$, so $R$ is a subring of $\mathcal{R}$. The nonzero codewords of the cyclic code of length $n$ generated by $h(x)$ is called {\bf constant vector}.
\begin{lem}
If the nonzero vector $z=(e,f)$ with $e$ not a constant vector, then there are at most $q^{n+1}$ generators $(1,a)$ such that $z\in C_a$ and $C_a$ is a double circulant code over $R.$
\end{lem}
\begin{proof}
By the Chinese Remainder Theorem (CRT), $(e,f)\simeq(e_1,f_1)\oplus(e_2,f_2)$. Since $(e,f)\in C_a$, then $f=ea$, $f_1=e_1a_1$ and $f_2=e_2a_2$, where $e_1$, $f_1$, $a_1\in R$ and $e_2$, $f_2$, $a_2\in \mathcal{R}$. Let $a_1=a_1^{\prime}+ua_1^{\prime\prime}$, $a_2=a_2^{\prime}+ua_2^{\prime\prime}$, where $a_1^{\prime}, a_1^{\prime\prime}\in \F_q$ and $a_2^{\prime}, a_2^{\prime\prime}\in\F_{q^{n-1}}$.

In the first constituent of $C_a$, we discuss on the unit character of $e_1$ as follows.
 \begin{itemize}
\item If $e_1\in R^*$, there exists only one solution for $a_1=\frac{f_1}{e_1}$.

\item If $e_1\in (u)\backslash\{0\}$, then $e_1=ue_1^{\prime}$ and $f_1=uf_1^{\prime}$, where $e_1^{\prime}\in\F_q^*,$ $f_1^{\prime}\in\F_q$. Since $u^2=u$, then $f_1=uf_1^{\prime}=ue_1^{\prime}a_1=ue_1^{\prime}(a_1^{\prime}+ua_1^{\prime\prime})=ue_1^{\prime}a_1^{\prime}+ue_1^{\prime}a_1^{\prime\prime}
    \Longleftrightarrow\frac{f_1^{\prime}}{e_1^{\prime}}=a_1^{\prime}+a_1^{\prime\prime}$. There are $q$ choices for $a_1$.

\item If $e_1\in (u-1)\backslash\{0\}$, then $e_1=(u-1){e_1}^{\prime}$ and $f_1=(u-1)f_1^{\prime}$, where $e_1^{\prime}\in\F_q^*,$ $f_1^{\prime}\in\F_q$. Since $u^2=u$, then $f_1=(u-1)f_1^{\prime}=(u-1)e_1^{\prime}a_1=(u-1)e_1^{\prime}(a_1^{\prime}+ua_1^{\prime\prime})
    =(u-1)e_1^{\prime}a_1^{\prime}
    \Longleftrightarrow\frac{f_1^{\prime}}{e_1^{\prime}}=a_1^{\prime}$ and $a_1^{\prime\prime}$ is arbitrary in $\F_q$. There are $q$ choices for $a_1$.
\item If $e_1$=0, then $a_1$ is arbitrary in $R$, there are $q^2$ choices for $a_1$.
\end{itemize}
In the second constituent of $C_a$, we discuss on the unit character of $e_2$ as follows.
 \begin{itemize}
\item If $e_2\in \mathcal{R}^*$, there exists only one solution of $a_2$.

\item If $e_2\in (u)\backslash\{0\}$, then $e_2=ue_2^{\prime}$ and $f_2=uf_2^{\prime}$, where $e_2^{\prime}\in\F_{q^{n-1}}^*,$ $f_2^{\prime}\in\F_{q^{n-1}}$. Since $u^2=u$, then $f_2=uf_2^{\prime}=ue_2^{\prime}a_2=ue_2^{\prime}(a_2^{\prime}+ua_2^{\prime\prime})=ue_2^{\prime}a_2^{\prime}+ue_2^{\prime}a_2^{\prime\prime}
    \Longleftrightarrow\frac{f_2^{\prime}}{e_2^{\prime}}=a_2^{\prime}+a_2^{\prime\prime}$. There are $q^{n-1}$ choices for $a_2$.

\item If $e_2\in (u-1)\backslash\{0\}$, then $e_2=(u-1){e_2}^{\prime}$ and $f_2=(u-1)f_2^{\prime}$, where $e_2^{\prime}\in\F_{q^{n-1}}^*,$ $f_2^{\prime}\in\F_{q^{n-1}}$. Since $u^2=u$, then $f_2=(u-1)f_2^{\prime}=(u-1)e_2^{\prime}a_2=(u-1)e_2^{\prime}(a_2^{\prime}+ua_2^{\prime\prime})
    =(u-1)e_2^{\prime}a_2^{\prime}
    \Longleftrightarrow\frac{f_2^{\prime}}{e_2^{\prime}}=a_2^{\prime}$ and $a_2^{\prime\prime}$ is arbitrary in $\F_{q^{n-1}}$. There are $q^{n-1}$ choices for $a_2$.
\item If $e_2=0$, then $e\equiv0\mod h(x)$, then $e$ is a constant vector, contradiction.
\end{itemize}
Thus, there are at most $q^{n+1}$ generators $(1, a)$ such that $z\in C_a$.
\end{proof}
\begin{lem}
If  $z=(e,f)\in R^{2n}$ with $e$ is not a constant vector,
then there are at most $4(1+q^{\frac{n-1}{2}})$ generators $(1,a)$ such that $z\in C_a$ and $C_a$ is a self-dual double circulant code over $R.$
\end{lem}
\begin{proof}
Keep the same notations as Lemma 4.1. In the first constituent of $C_a$, there are at most $4$ generators $(1,a_1)$ such that $C_1$ is self-dual double circulant code over $R$ due to Theorem 3.1.

In the second constituent of $C_a$, let $a_2=a_2^{\prime}+ua_2^{\prime\prime}$, where $a_2^{\prime},$ $a_2^{\prime\prime}\in\F_{q^{n-1}}$, we discuss on the unit character of $e_2$ as follows.
\begin{itemize}
\item If $e_2\in \mathcal{R}^*$, then there exists a unique solution of $a_2=\frac{f_2}{e_2}$.

\item If $e_2\in (u)\backslash\{0\}$, then $e_2=ue_2^{\prime}$, $f_2=uf_2^{\prime}$ and $$f_2=ua_2e_2^{\prime}=u(a_2^{\prime}+ua_2^{\prime\prime})e_2^{\prime}=u(a_2^{\prime}+a_2^{\prime\prime})e_2^{\prime}\Longleftrightarrow a_2^{\prime}+a_2^{\prime\prime}=\frac{f_2^{\prime}}{e_2^{\prime}},$$
     where $e_2^{\prime}\in \F^*_{q^{n-1}}$ and $f_2^{\prime}\in \F_{q^{n-1}}.$ Since $C_a$ is a self-dual double circulant code, then $(1,a_2)\cdot(1,a_2)=1+a_2\overline{a_2}=1+a_2{a_2}^{q^{\frac{n-1}{2}}}=0$.
         This is equivalent to
        \begin{equation*}\label{den1}
 \small
\begin{cases}
 \emph{ }a_2^{\prime}{a_2^{\prime}}^{q^{\frac{n-1}{2}}}=-1,\\
   \emph{ }(a_2^{\prime}+a_2^{\prime\prime})({a_2^{\prime}}^{q^{\frac{n-1}{2}}}+{a_2^{\prime\prime}}^{q^{\frac{n-1}{2}}})=-1.\\
\end{cases}\Longleftrightarrow
\begin{cases}
 \emph{ }Norm(a_2^{\prime})=-1,\\
   \emph{ }Norm(a_2^{\prime}+a_2^{\prime\prime})=-1,\\
\end{cases}
\end{equation*}
where $Norm$ is a map from $\F_{q^{n-1}}$ to $\F_{q^{\frac{n-1}{2}}}$. So there are $q^{\frac{n-1}{2}}+1$ roots for $Norm(a_2^{\prime})=-1$ and $q^{\frac{n-1}{2}}+1$ roots for $Norm(a_2^{\prime}+a_2^{\prime\prime})=-1$. But $a_2^{\prime\prime}$ is determined by $a_2^{\prime}$, then there are at most $1+q^{\frac{n-1}{2}}$ choices for $a_2^{\prime}$, hence for $a_2$.
\item If $e_2\in (u-1)\backslash\{0\}$, then $e_2=(u-1)e_2^{\prime}$ and $f_2=(u-1)f_2^{\prime}$, where $e_2^{\prime}\in \F^*_{q^{n-1}}$ and $f_2^{\prime}\in \F_{q^{n-1}}$, which implies that $$f_2=(u-1)a_2e_2^{\prime}=(u-1)(a_2^{\prime}+ua_2^{\prime\prime})e_2^{\prime}\Longleftrightarrow a_2^{\prime}=\frac{f_2^{\prime}}{e_2^{\prime}}.$$ Since $C_a$ is a self-dual double circulant code, using the same procedure as above, we have
\begin{equation*}\label{den1}
 \small
\begin{cases}
 \emph{ }Norm(a_2^{\prime})=-1,\\
   \emph{ }Norm(a_2^{\prime}+a_2^{\prime\prime})=-1.\\
\end{cases}
\end{equation*}
Hence, there are $1+q^{\frac{n-1}{2}}$ choices for $a_2$ because $a_2^{\prime}=\frac{f_2^{\prime}}{e_2^{\prime}}$ and $1+q^{\frac{n-1}{2}}$ roots for $a_2^{\prime\prime}$.
 \item If $e_2=0$, then $e\equiv0\mod h(x)$, then $e$ is a constant vector, contradiction.
\end{itemize}

 So there are at most $4(1+q^{\frac{n-1}{2}})$ generators $(1,a)$ such that $z\in C_a.$ The proof is done.
\end{proof}

In number theory, Artin's conjecture on primitive roots states that a given integer $q$ which is neither a perfect square nor $-1$ is a primitive root modulo
infinitely many primes $l$ \cite{M}. This was proved conditionally under the Generalized Riemann Hypothesis (GRH) by Hooley \cite{H}. Hence we can get infinite families
of double circulant codes $C(2n)$ over $R$ where the analysis made for $x^n-1$ with only two irreducible factors applies.

Recall the $q$-ary entropy function defined for $0\leq t\leq\frac{q-1}{q}$ by
\begin{equation*}\label{den1}
H_q(t)=\begin{cases}
 \emph{ }0,  ~~~~~~~~~~~~~~~~~~~~~~~~~~~~~~~~~~~~~~~~~~~~~~~~~~~~~~~~~~~~~~~~{\rm{if}}~~ t=0,\\
   \emph{ }t{\rm{log}}_q(q-1)-t{\rm{log}}_q(t)-(1-t){\rm{log}}_q(1-t), ~~~~~~~~~~~{\rm{if}}~~0<t\leq\frac{q-1}{q}. \\
\end{cases}
\end{equation*}
This quantity is instrumental in the estimation of the volume of high-dimensional Hamming balls when the base field is $\mathbb{F}_q$.
The result we are using is that the volume of the Hamming ball of radius $tn$ is asymptotically equivalent, up to subexponential terms, to $q^{nH_q(t)}$, when
$0<t<1$, and $n$ goes to infinity \cite[Lemma 2.10.3]{W}.
The main  result obtained in this paper is as follows.

\begin{thm}\label{lem2}
Let $n$ be an odd prime, $n>q,$ and $q$ be a primitive root modulo $n.$ The family of Gray images of self-dual double  circulant codes over $R$ of length $2n,$ of relative distance $\delta,$ and rate $1/2,$
satisfies $H_q(\delta)\geq \frac{1}{8}$.
The family of Gray images of LCD double  circulant codes over $R$ of length $2n,$ of relative distance $\delta,$ and rate $1/2,$
satisfies $H_q(\delta)\geq \frac{1}{4}$. In particular, both families codes are good.
\end{thm}\begin{proof}
Let $\Omega_n$ denote the size of the family. Thus, for $n\rightarrow \infty,$ we have, by Theorem 3.1, $\Omega_n \sim 4 q^{n-1}$ for self-dual double circulant codes, and by Theorem 3.2,
$\Omega_n \sim q^{2n-2}$ for LCD double circulant codes.
Assume we can prove that for $n$ large enough
$\Omega_n>\lambda_n B(d_n)$, where $B(r)$ denotes the number of vectors in $R^{2n}$ with Hamming weight of their $\F_q$ image $<r.$
Here $\lambda_n=4(1+q^{(n-1)/2})$ for self-dual codes, and $\lambda_n=q^{n+1}$ for LCD codes.

This would imply, by Lemmas 4.1 and 4.2, that there are codes of length $2n$ in the family with minimum Hamming distance of their $\F_q$ image $\ge d_n.$ Denote by
$\delta$ the relative distance of this family
of $q$-ary codes.

If we take $d_n$ the largest number satisfying $\Omega_n>\lambda_n B(d_n)$, and assume a growth of the form
$d_n\sim 4 \delta_0 n,$ and thus $\Omega_n \sim \lambda_n B(d_n),$ for $n \rightarrow \infty$ then, using an entropic estimate for $B(d_n)\sim q^{4nH_q(\delta_0)}$ \cite[Lemma 2.10.3]{W} yields, with the said values of $\Omega_n$ and $\lambda_n$ the estimate
  $H_q(\delta_0)=\frac{1}{8}$ for self-dual codes and $H_q(\delta_0)=\frac{1}{4}$ for LCD codes. The result follows by observing that, by the definition of $\delta,$ we have $\delta\ge \delta_0.$
\end{proof}

\subsection{\textbf{ Distance Bound for Double Negacirculant Codes}}
\subsubsection{\textbf{Factorization of $x^n+1$}}
In order to describe the factorization of $x^n+1$ over $R$, where $n=2^a$, we first recall the definition of {\bf Dickson polynomials}.
For $\alpha\in R$, we define the Dickson polynomial of paramater $\alpha$ and degree $n$ as
$$D_n(x,\alpha)=\sum\limits_{j=0}^{\lfloor n/2\rfloor}\frac{n}{n-j}{n-j \choose j}(-\alpha)^jx^{n-2j}.$$
The complete factorization of $x^{2^n}+1$ over $\mathbb{F}_q$ with $q\equiv 3~(\rm{mod}~4)$ is given in the following theorem \cite{HM}.
\begin{thm}\label{tem}\cite{HM} Let $q\equiv 3~({mod}~4)$, where $q=2^Am-1, A\geq 2$ and $m$ is an odd integer. Let $n\geq2$,
\begin{enumerate}
  \item [(a)] if $n<A,$ then $x^{2^n}+1$ is the product of $2^{n-1}$ irreducible trinomials over $\mathbb{F}_q$ $$x^{2^n}+1=\prod\limits_{\gamma\in\Gamma}(x^2+\gamma x+1),$$ where $\Gamma$ is the set of all roots of $D_{2^{n-1}}(x,1)$.
  \item [(b)] if $n\geq A,$ then $x^{2^n}+1$ is the product of $2^{A-1}$ irreducible trinomials over $\mathbb{F}_q$ $$x^{2^n}+1=\prod\limits_{\delta\in\triangle}(x^{2^{n-A+1}}+\delta x^{2^{n-A}}-1),$$ where $\triangle$ is the set of all roots of $D_{2^{A-1}}(x,-1)$.
\end{enumerate}
\end{thm}
\begin{ex} If $q=3$, i.e., $q\equiv 3~(mod~4)$, then $q=2^2\cdot 1-1$. This implies $A=2,m=1,$ and $D_2(x,-1)=0$, i.e., $\triangle=\{1,2\}$. By Theorem 4.4 $$x^{2^n}+1=(x^{2^{n-1}}+ x^{2^{n-2}}-1)(x^{2^{n-1}}+2x^{2^{n-2}}-1).$$
\end{ex}
\begin{ex} If $q=11$, i.e., $q\equiv 3~(mod~4)$, then $q=2^2\cdot 3-1$. This implies $A=2,m=3,$ and $D_2(x,-1)=0$, i.e., $\triangle=\{3,8\}$. By Theorem 4.4 $$x^{2^n}+1=(x^{2^{n-1}}+3 x^{2^{n-2}}-1)(x^{2^{n-1}}+8x^{2^{n-2}}-1).$$
\end{ex}
\begin{ex} If $q=7$, i.e., $q\equiv 3~(mod~4)$, then $q=2^3\cdot 1-1$. This implies $A=3,m=1,$ and $D_4(x,-1)=0$, i.e., $\triangle=\{1,3,4,6\}$. By Theorem 4.4 $$x^{2^n}+1=(x^{2^{n-2}}+ x^{2^{n-3}}-1)(x^{2^{n-2}}+3x^{2^{n-3}}-1)(x^{2^{n-2}}+4x^{2^{n-3}}-1)(x^{2^{n-2}}+6x^{2^{n-3}}-1).$$
\end{ex}
Next, we need the analogous factorization theorem when $q\equiv 1~(\rm{mod}~4)$ as follows.
\begin{thm}\label{tem1}\cite{A2} Let $q\equiv 1~({mod}~4)$, where $q=2^{A+1}m+1, A\geq 1, m$ is an odd integer. Denote the set of all primitive $2^k$-th roots of unity in $\mathbb{F}_q$ by $U_k$. If $n\geq2$, then
\begin{enumerate}
  \item [(a)] if $n\leq A,$ then ord$_{2^{n+1}}(q)=1$, $x^{2^n}+1$ is the product of $2^{n}$ linear factors over $\mathbb{F}_q$ $$x^{2^n}+1=\prod\limits_{u\in U_{n+1}}(x+u).$$
  \item [(b)] if $n\geq A+1,$ then ord$_{2^{n+1}}(q)=2^{n-A}$, $x^{2^n}+1$ is the product of $2^{A}$ irreducible binomials over $\mathbb{F}_q$ of degree $2^{n-A}$ $$x^{2^n}+1=\prod\limits_{u\in U_{A+1}}(x^{2^{n-A}}+u).$$
\end{enumerate}
\end{thm}
\begin{ex} If $q=5$, i.e., $q\equiv 1~(mod~4)$, then $q=2^2\cdot 1+1$. This implies $A=1,m=1,$ and $U_2=\{2,3\}$. By Theorem \ref{tem1} $$x^{2^n}+1=(x^{2^{n-1}}+2)(x^{2^{n-1}}+3).$$
\end{ex}
\begin{ex} If $q=13$, i.e., $q\equiv 1~(mod~4)$, then $q=2^2\cdot 3+1$. This implies $A=1,m=3,$ and $U_2=\{5,8\}$. By Theorem \ref{tem1} $$x^{2^n}+1=(x^{2^{n-1}}+5)(x^{2^{n-1}}+8).$$
\end{ex}
\begin{ex} If $q=41$, i.e., $q\equiv 1~(mod~4)$, then $q=2^3\cdot 5+1$. This implies $A=2,m=5,$ and $U_3=\{3,14,27,38\}$. By Theorem \ref{tem1} $$x^{2^n}+1=(x^{2^{n-2}}+3)(x^{2^{n-2}}+14)(x^{2^{n-2}}+27)(x^{2^{n-2}}+38).$$
\end{ex}

{\bf Remark:} There is an error in the case (b) of Theorem 4 in \cite{A}, i.e., ``$x^{2^n}+1=\prod\limits_{\gamma\in\triangle}(x^{{n-A+1}}+\delta x^{{n-A}}-1)$" should be changed to ``$x^{2^n}+1=\prod\limits_{\gamma\in\triangle}(x^{2^{n-A+1}}+\delta x^{2^{n-A}}-1)$". According to \cite{HM}, there is another error of Theorem 6 in \cite{A}, ``$q=2^{A}m+1, A\geq 2$" should be changed to ``$q=2^{A+1}m+1, A\geq 1$".

So we can cast the factorization $x^n+1$ over $\F_q$ into two or four irreducible polynomials by limiting the size of $\triangle$ and $U$ in Theorems 4.4 and 4.8.
This factorization carries over $R$
because $\F_q$ is subring of $R$.

The polynomial $x^n+1$ factors into two irreducible polynomials that are reciprocal of each other over $R$, this is the case if $q\equiv \pm 1$~(mod~$4)$, where $q=2^2m\pm 1, m$ odd, it happens if $q=3, 5, 11, 13, 19, 27, 29, 37,$ etc.
The polynomial $x^n+1$ factors into four irreducible polynomials which are pairwise reciprocal of each other over $R$, this is the case if $q\equiv \pm 1$~(mod~$4)$, where $q=2^3m\pm 1, m$ odd, it happens if $q=7, 23, 25, 41,$ etc.
\subsubsection{\textbf{Distance Bounds for Decomposition I and II }}
\noindent\textbf{Decomposition I:} If $x^n+1=h(x)h^*(x)$, where $h(x)$ and $h^*(x)$ are irreducible polynomials and reciprocal of each other,
where $\deg (h(x))=\frac{n}{2}$, then $R[x]/(h(x))\simeq R[x]/(h^*(x))\simeq\F_{q^{\frac{n}{2}}}+u\F_{q^{\frac{n}{2}}}$. For convenience,
$K_1=R[x]/(h(x))$ and $K_2=R[x]/(h^*(x))$. We have provided several examples of that situation, such as Examples 4.5, 4.6, 4.9 and 4.10.
\begin{lem}\label{lem1} Let $q$ be odd, $n$ be a power of $2$, $x^n+1$ with \textbf{Decomposition I}.
If $z=(e,f)\in R^{2n}$ is a nonzero vector, then there are at most $q^\frac{3n}{2}$ generators $(1,a)$ such that $z\in C_a$ and $C_a$ is a double negacirculant code over $R$.
\end{lem}
\begin{proof}
By the Chinese Remainder Theorem (CRT), $(e,f)=(e_1,f_1)\oplus(e_2,f_2)$. Since $(e,f)\in C_a$, then $f=ea$, $f_1=e_1a_1$ and $f_2=e_2a_2$, where $e_1$, $f_1$, $a_1\in K_1$ and $e_2$, $f_2$, $a_2\in K_2$. Let $a_1=a_1^{\prime}+ua_1^{\prime\prime}$, $a_2=a_2^{\prime}+ua_2^{\prime\prime}$, where $a_1^{\prime}$, $a_1^{\prime\prime}$, $a_2^{\prime},$ $a_2^{\prime\prime}\in\F_{q^{\frac{n}{2}}}$.

In the first constituent of $C_a$, we discuss on the unit character of $e_1$ as follows.
 \begin{itemize}
\item If $e_1\in K_1^*$, there exists only one solution for $a_1=\frac{f_1}{e_1}$.

\item If $e_1\in (u)\backslash\{0\}$, then $e_1=ue_1^{\prime}$ and $f_1=uf_1^{\prime}$, where $e_1^{\prime}\in\F^*_{q^{\frac{n}{2}}},$ $f_1^{\prime}\in\F_{q^{\frac{n}{2}}}$. Since $u^2=u$, then $f_1=uf_1^{\prime}=ue_1^{\prime}a_1=ue_1^{\prime}(a_1^{\prime}+ua_1^{\prime\prime})=ue_1^{\prime}a_1^{\prime}+ue_1^{\prime}a_1^{\prime\prime}
    \Longleftrightarrow\frac{f_1^{\prime}}{e_1^{\prime}}=a_1^{\prime}+a_1^{\prime\prime}$. There are $q^{\frac{n}{2}}$ choices for $a_1$.

\item If $e_1\in (u-1)\backslash\{0\}$, then $e_1=(u-1){e_1}^{\prime}$ and $f_1=(u-1)f_1^{\prime}$, where $e_1^{\prime}\in\F^*_{q^{\frac{n}{2}}}$, $f_1^{\prime}\in\F_{q^{\frac{n}{2}}}$. Since $u^2=u$, then $f_1=(u-1)f_1^{\prime}=(u-1)e_1^{\prime}a_1=(u-1)e_1^{\prime}(a_1^{\prime}+ua_1^{\prime\prime})
    =(u-1)e_1^{\prime}a_1^{\prime}
    \Longleftrightarrow\frac{f_1^{\prime}}{e_1^{\prime}}=a_1^{\prime}$ and $a_1^{\prime\prime}$ is arbitrary in $\F_q$. There are ${q^{\frac{n}{2}}}$ choices for $a_1$.
\item If $e_1=0$, then $a_1$ is arbitrary in $K_1$, there are $q^n$ choices for $a_1.$
\end{itemize}
Using the same argument as above in the second constituent of $C_a$, there are also at most $q^{n}$ choices for $a_2$. But if $z$ is not zero, then $e_1$ and $e_2$ cannot both be zero. Hence there are only at most $q^{\frac{3n}{2}}$ generators $(1,a)$ such that $z\in C_a.$
We have thus proved the lemma.
\end{proof}
\begin{lem}\label{lem1} Let $q$ be odd, $n$ be a power of $2.$ Assume $x^n+1$ satisfies decomposition I.
 If $z=(e,f)\in R^{2n}$ is a nonzero vector,
then there are at most $q^{\frac{n}{2}}$ generators $(1,a)$ such that $z\in C_a$ and that $C_a$ is self-dual double negacirculant code of length $2n$ over $R$.
\end{lem}
\begin{proof}
Keep the same notations as Lemma 4.12. Since $C_a$ is a self-dual code, then $\langle(1,a_1), (1,\\a_2)\rangle=1+a_1a_2=0$. Clearly, $a_1,a_2$ are units and $a_2$ is determined by $a_1$. Since $e_1$ and $e_2$ can not both be zero, without loss of generality, we may assume $e_1\neq0$, then we just need to consider the first constituent of $C_a$ by symmetry. Next, we discuss on the unit character of $e_1$ as follows.
 \begin{itemize}
\item If $e_1\in K_1^*$, there exists only one solution of $a_1=\frac{f_1}{e_1}$.

\item If $e_1\in (u)\backslash\{0\}$, then $e_1=ue_1^{\prime}$ and $f_1=uf_1^{\prime}$, where $e_1^{\prime}\in\F^{*}_{q^{\frac{n}{2}}}$ and $f_1^{\prime}\in\F_{q^{\frac{n}{2}}}$. Since $u^2=u$, then $f_1=uf_1^{\prime}=ue_1^{\prime}a_1=ue_1^{\prime}(a_1^{\prime}+ua_1^{\prime\prime})=ue_1^{\prime}a_1^{\prime}+ue_1^{\prime}a_1^{\prime\prime}
    \Longleftrightarrow\frac{f_1^{\prime}}{e_1^{\prime}}=a_1^{\prime}+a_1^{\prime\prime}$. There are $q^{\frac{n}{2}}$ choices for $a_1$.

\item If $e_1\in (u-1)\backslash\{0\}$, then $e_1=(u-1){e_1}^{\prime}$ and $f_1=(u-1)f_1^{\prime}$, where $e_1^{\prime}\in\F^*_{q^{\frac{n}{2}}}$ and $f_1^{\prime}\in\F_{q^{\frac{n}{2}}}$. Since $u^2=u$, then $f_1=(u-1)f_1^{\prime}=(u-1)e_1^{\prime}a_1=(u-1)e_1^{\prime}(a_1^{\prime}+ua_1^{\prime\prime})
    =(u-1)e_1^{\prime}a_1^{\prime}
    \Longleftrightarrow\frac{f_1^{\prime}}{e_1^{\prime}}=a_1^{\prime}$ and $a_1^{\prime\prime}$ is arbitrary in $\F_q$. There are ${q^{\frac{n}{2}}}$ choices for $a_1$.
\end{itemize}
 Hence there are at most $q^{\frac{n}{2}}$ generators $(1,a)$ such that $z\in C_a.$ We have thus proved the lemma.
\end{proof}
\noindent\textbf{Decomposition II:} If $x^n+1=u_1(x)u_1^*(x)u_2(x)u_2^*(x)$, where $u_1(x)~(resp.~u_2(x))$ and $u_1^*(x)~(resp.~u^*_2(x))$ are irreducible polynomials
and reciprocal of each other, $\deg (u_i(x))=q^{\frac{n}{4}}, i=\{1,2\}$, then
$R[x]/(u_1(x))\simeq R[x]/(u_1^*(x))\simeq R[x]/(u_2(x))\simeq R[x]/(u_2^*(x))\simeq\F_{q^{\frac{n}{4}}}+u\F_{q^{\frac{n}{4}}}$.
For convenience, let $I_1=R[x]/(u_1(x))$, $I_2=R[x]/(u_1^*(x))$, $I_3=R[x]/(u_2(x))$ and $I_4=R[x]/(u_2^*(x))$. We have provided several examples of that situation, such as Examples 4.7 and 4.11.
\begin{lem}\label{lem1} Let $q$ be odd, $n$ be a power of $2$, and $x^n+1$ satisfying Decomposition II .
If $z=(e,f)\in R^{2n}$ is a nonzero vector, then there are at most $q^\frac{7n}{4}$ generators $(1,a)$ such
that $z\in C_a$ and $C_a$ is a double negacirculant code of length $2n$ over $R$.
\end{lem}
\begin{proof}
By the Chinese Remainder Theorem (CRT), $(e,f)=(e_1,f_1)\oplus(e_2,f_2)\oplus(e_3,f_3)\oplus(e_4,f_4)$. Since $(e,f)\in C_a$, then $f=ea$, $f_1=e_1a_1$, $f_2=e_2a_2$, $f_3=e_3a_3$ and $f_4=e_4a_4$, where $e_1$, $f_1$, $a_1\in I_1$, $e_2$, $f_2$, $a_2\in I_2$, $e_3$, $f_3$, $a_3\in I_3$ and $e_4$, $f_4$, $a_4\in I_4$.

In the first constituent of $C_a$, Let $a_1=a_1^{\prime}+ua_1^{\prime\prime}$, where $a_1^{\prime}$, $a_1^{\prime\prime}\in\F_{q^{\frac{n}{4}}}$, we discuss on the unit character of $e_1$ as follows.
 \begin{itemize}
\item If $e_1\in I_1^*$, then there exists only one solution for $a_1$ that is $a_1=\frac{f_1}{e_1}$.

\item If $e_1\in (u)\backslash\{0\}$, then $e_1=ue_1^{\prime}$ and $f_1=uf_1^{\prime}$, where $e_1^{\prime}\in\F^*_{q^{\frac{n}{4}}}$ and $f_1^{\prime}\in\F_{q^{\frac{n}{4}}}$. Since $u^2=u$, then $f_1=uf_1^{\prime}=ue_1^{\prime}a_1=ue_1^{\prime}(a_1^{\prime}+ua_1^{\prime\prime})=ue_1^{\prime}a_1^{\prime}+ue_1^{\prime}a_1^{\prime\prime}
    \Longleftrightarrow\frac{f_1^{\prime}}{e_1^{\prime}}=a_1^{\prime}+a_1^{\prime\prime}$. There are $q^{\frac{n}{4}}$ choices for $a_1$.

\item If $e_1\in (u-1)\backslash\{0\}$, then $e_1=(u-1){e_1}^{\prime}$ and $f_1=(u-1)f_1^{\prime}$, where $e_1^{\prime}\in\F^*_{q^{\frac{n}{4}}}$ and $f_1^{\prime}\in\F_{q^{\frac{n}{4}}}$. Since $u^2=u$, then $f_1=(u-1)f_1^{\prime}=(u-1)e_1^{\prime}a_1=(u-1)e_1^{\prime}(a_1^{\prime}+ua_1^{\prime\prime})
    =(u-1)e_1^{\prime}a_1^{\prime}
    \Longleftrightarrow\frac{f_1^{\prime}}{e_1^{\prime}}=a_1^{\prime}$, while $a_1^{\prime\prime}$ is arbitrary in $\F_q$. Thus there are ${q^{\frac{n}{4}}}$ choices for $a_1$.
\item If $e_1=0$, then $a_1$ is arbitrary in $I_1$, there are $q^{\frac{n}{2}}$ choices for $a_1.$
\end{itemize}
Using the same argument as above in the three remaining constituents of $C_a$. There are also at most $q^{\frac{n}{2}}$ choices for $a_i$, where $i=2,3,4.$ But $z$ is not zero, then $e_1$, $e_2$, $e_3$ and $e_4$ can not be zero at the same time. Hence there are at most $q^{\frac{3n}{2}}\times q^{\frac{n}{4}}$ generators $(1,a)$ such that $z\in C_a.$ We have thus proved the lemma.
\end{proof}
\begin{lem}\label{lem1} Let $q$ be odd, $n$ be a power of $2$, and $x^n+1$ satisfying Decomposition II.
If $z=(e,f)\in R^{2n}$ is a nonzero vector, then there are at most $q^{\frac{3n}{4}}$ generators $(1,a)$
such that $z\in C_a,$ and $C_a$ is a self-dual double negacirculant code over $R$.
\end{lem}
\begin{proof}Keep the same notations as Lemma 4.14. Since $C_a$ is a self-dual code, we then have
\begin{equation*}\label{den1}
 \small
 \begin{cases}
 \emph{ }\langle(1,a_1), (1,a_2)\rangle=0,\\
   \emph{ }\langle(1,a_3),(1,a_4)\rangle=0.\\
\end{cases}\Longleftrightarrow
\begin{cases}
 \emph{ }1+a_1a_2=0,\\
   \emph{ }1+a_3a_4=0.\\
\end{cases}
\end{equation*}So we just need to consider $a_1$ and $a_3$, because $a_2$ and $a_4$ are determined by $a_1$ and $a_3$, respectively. Since $z$ is not zero, then we may assume $e_1\neq0$.
Using a similar argument as in the proof of Lemma 4.13, we see that there are only at most $q^{\frac{n}{4}}$ choices for $a_1$ and $q^{\frac{n}{2}}$ choices for $a_3$. Hence, there are at most $q^{\frac{3n}{4}}$ generators such that $z\in C_a.$
\end{proof}
The following results can be proved by the same method as employed in the last subsection.
\begin{thm} Under the condition of the \textbf{Decomposition I} ($resp.$ \textbf{Decomposition II}), if $q$ is an odd prime power, and $n$ is a power of $2$, then the family of Gray images of self-dual double  negacirculant codes over $R$ of length $2n,$ of relative distance $\delta,$ and rate $1/2,$
satisfies $H_q(\delta)\geq \frac{1}{8}$ ($resp.$ $H_q(\delta)\geq \frac{1}{16}$). In particular, these families of codes are good.
\end{thm}
\begin{proof}
The proof follows the method of the proof of Theorem 4.3 with $\Omega_n\sim q^{n}$ for both Decomposition I and Decomposition II, by Theorem 3.4 and $\lambda_n=q^{\frac{n}{2}}$ and $\lambda_n=q^{\frac{3n}{4}}$ by Lemmas 4.13 and 4.15, respectively. The details are omitted.
\end{proof}
\begin{thm} Under the condition of the \textbf{Decomposition I} ($resp.$ \textbf{Decomposition II}), if $q$ is an odd prime power, and $n$ is a power of $2$, then the family of Gray images of LCD double negacirculant codes over $R$ of length $2n,$ of relative distance $\delta,$ and rate $1/2,$
satisfies $H_q(\delta)\geq \frac{1}{8}$ ($resp.$ $H_q(\delta)\geq \frac{1}{16}$). In particular, these families of codes are good.
\end{thm}
\begin{proof}
The proof follows the method of the proof of Theorem 4.3 with $\Omega_n\sim q^{2n}$ for both Decomposition I and Decomposition II, by Theorem 3.4 and $\lambda_n=q^\frac{3n}{2}$ and $\lambda_n=q^\frac{7n}{4}$ by Lemmas 4.12 and 4.14, respectively. The details are omitted.
\end{proof}
\section{Numerical Examples}

We need the following lemma, before doing any computation.
{\lem If $C$ has generator matrix $G=(I,A),$ with $A=A_1+uA_2,$ where $A_1,A_2$ are $q$-ary matrices of order $n$, and $I$ denotes the identity matrix of order $n,$ then $\phi(C)$ has generator matrix
$$\begin{pmatrix} -I&I&-A_1-A_2& A_1+A_2\\0&2I&-A_2&2A_1+A_2 \end{pmatrix}.$$

}

\begin{proof}
The first row of the first matrix is $\phi(uG),$ and its second row is $\phi(G).$ The matrix spans subsets of $\phi(C)$. Since the rank is $2n,$ the result follows.
\end{proof}
In Table \ref{Table:1} and Table \ref{Table:2}, we have provided some examples of LCD and self-dual double circulant codes respectively with the best parameters obtained by Magma search and Lemma 5.1. The coefficients of degree $n$ polynomial $a_1(x)$ and $a_2(x)$ in $\F_5[x]$ are written in decreasing powers of $x.$ For example for
$n=3,$ the entry $311$ means $3x^2+x+1.$
The parameters over $\F_5$ are given in the form $[4n,2n,d]$ where $d$ is the minimum distance.
The entry in the rightmost column is the best known distance of a $[4n,2n]$ (self-dual) code over $\F_5$, obtained by looking up the tables in \cite{WWW0} and \cite{WWW}. When the distance of the code constructed reaches that value the parameters are starred.
\begin{table}\caption{ Gray image of LCD double circulant codes over $\F_5+u\F_5$, $^*$: optimal codes}
$$
\begin{array}{|c|l|l|c|c|}
\hline
n&\textbf{$a_1(x)$}&\textbf{$a_2(x)$}&\textbf{Parameters over }{ \F}_5&\textbf{Distance \cite{WWW0}}\\
\hline
2&
4 0&
4 2&
[8, 4, 4]^*&4\\
3&
1 2 1&
4 0 2&
[12, 6, 6]^*&6\\

4&
0 3 3 4&
3 2 4 2&
[16, 8, 6]&7\\

5&
4 3 0 3 0&
0 4 1 3 1&
[20, 10, 8]^*&8\\
6&
0 1 0 0 4 4&
1 3 2 2 0 2&
[24, 12, 8]&9\\
7&
1 4 0 2 1 2 4&
2 1 1 3 4 2 4&
[28, 14, 10]&11\\
8&
3 4 4 3 0 1 1 0&
2 4 0 2 3 1 2 1&
[32, 16, 11]^*&11\\
9&
0 3 3 3 0 2 1 2 2&
3 1 4 3 2 1 0 0 0&
[36, 18, 12]^*&12\\

\hline
\end{array}
$$\label{Table:1}
\end{table}

\begin{table}\caption{ Gray image of self-dual double circulant codes over $\F_5+u\F_5$, $^*$: optimal codes}
$$
\begin{array}{|c|l|l|c|c|}
\hline
n&\textbf{$a_1(x)$}&\textbf{$a_2(x)$}&\textbf{Parameters over }{ \F}_5&\textbf{Distance \cite{WWW}}\\
\hline
2&
2 0&
3 2&
[8, 4, 4]^*&4\\
3&
3 1 3&
2 4 0&
[12, 6, 4]&6\\

4&
4 1 4 4&
2 0 2 0&
[16, 8, 4]&7\\

5&
2 2 3 1 4&
1 2 0 0 3&
[20, 10, 8]^*&8\\
6&
0 3 1 2 4 3&
0 3 2 0 0 4&
[24, 12, 8]&9\\
7&
2 2 2 0 2 2 2&
1 2 1 4 2 0 1&
[28, 14, 8]&10\\
8&
0 3 3 1 0 0 2 4&
2 4 0 2 2 2 1 2&
[32, 16,8]&10\\
9&
3 3 1 4 2 1 0 0 3&
4 2 3 0 0 2 1 2 2&
[36, 18, 10]&12\\

\hline
\end{array}
$$\label{Table:2}
\end{table}

\section{\textbf{Conclusion and Open Problems}}

In this paper we have studied double circulant and double negacirculant codes over the ring $R=\F_q+u\F_q.$ It might be worth looking at quasi-cyclic codes of higher index, or more generally, at quasi-twisted codes
of higher index. Four-circulant or four-negacirculant codes are natural candidates for this exploration \cite{SZS}.

While passing from fields to rings increases significantly the complexity of the proofs and calculations, it might still
be worth looking at other polynomial rings with $u$ having a minimal polynomial of higher degree, or defined by several variables.

\end{document}